\documentclass{llncs}

\usepackage{amsmath}
\usepackage{amsfonts}
\usepackage{amssymb}
\usepackage{thmtools, thm-restate}

\usepackage{graphicx}
\usepackage{epic}
\usepackage{eepic}
\usepackage{epsfig,float}
\usepackage{pdfsync}
\usepackage{multicol}
\usepackage{pdfsync}
\usepackage{color}
\usepackage{verbatim}

\newcommand{\Sig}{\Sigma}
\newcommand{\eps}{\varepsilon}

\newcommand{\rev}{R}
\newcommand{\deter}{D}

\newcommand{\cA}{\mathcal{A}}

\newcommand{\bA}{\mathbf{A}}
\newcommand{\cD}{\mathcal{D}}

\newcommand{\cN}{\mathcal{N}}

\newcommand{\emp}{\emptyset}
\newcommand{\qedb}{\hfill$\blacksquare$} 
\newcommand{\ol}{\overline}

\newcommand{\im}{\operatorname{im}}
\newcommand{\coim}{\operatorname{coim}}
\newcommand{\preim}{\operatorname{preim}}

\newcommand{\noin}{\noindent}

\newcommand{\be}{\begin{enumerate}}
\newcommand{\ee}{\end{enumerate}}
\newcommand{\bi}{\begin{itemize}}
\newcommand{\ei}{\end{itemize}}

\title{Maximal Syntactic Complexity of Regular Languages Implies Maximal Quotient  Complexities of Atoms\thanks{This work was supported by the Natural Sciences and Engineering Research Council of Canada under grant No.~OGP0000871.}}
\author{Janusz Brzozowski and Gareth Davies}

\institute{David R. Cheriton School of Computer Science, University of Waterloo, \\
Waterloo, ON, Canada N2L 3G1\\
{\tt \{brzozo,gdavies\}@uwaterloo.ca}
}

\begin{document}
\maketitle

\begin{abstract}
We relate two measures of complexity of regular languages.
The first is syntactic complexity, that is, the cardinality of the syntactic semigroup of the language.
That semigroup is isomorphic to the semigroup of transformations of states induced by non-empty words in the minimal deterministic finite automaton accepting the language. If the language has $n$ left quotients (its minimal automaton has $n$ 
states), then its syntactic complexity is at most $n^n$ and this bound is tight.
The second measure consists of the  quotient (state) complexities of the atoms of the language, where atoms  
 are non-empty intersections of complemented and uncomplemented quotients.
A regular language has at most $2^n$ atoms and this bound is tight. 
The maximal quotient complexity of any atom  with $r$ complemented 
quotients is $2^n-1$, if $r=0$ or $r=n$, and 
$1+\sum_{k=1}^{r} \sum_{h=k+1}^{k+n-r} \binom{h}{n} \binom{k}{h}$,
otherwise. 
We prove that if a language has maximal syntactic complexity, then it has $2^n$ atoms
and each atom has maximal quotient complexity, but the converse is false.
\medskip

\noin

\end{abstract}

\section{Introduction}
In recent years much of the theory of the so-called \emph{descriptional complexity} of regular languages has been concerned with state complexity. The \emph{state complexity} of a regular language~\cite{Yu01} is the number of states in the minimal complete deterministic finite automaton (\emph{DFA}) recognizing the language. An equivalent notion is \emph{quotient complexity}~\cite{Brz10}, which is the number of left quotients of the language, where the \emph{left quotient} (or simply \emph{quotient}) of a language $L$ over an alphabet $\Sig$ by a word $w\in\Sig^*$ is $w^{-1}L=\{x\mid wx\in L\}$.
The \emph{(state/quotient) complexity of an operation} on regular languages is the maximal  complexity of the language resulting from the operation as a function of the  complexities of the arguments.
The operations considered may be \emph{basic}, for example, union, star or product (concatenation),
or \emph{combined}, for example, star of union or reversal of product.
Basic operations were first studied by Maslov~\cite{Mas70} in 1970, and later by Yu, Zhuang and K. Salomaa~\cite{YZS94} in 1994. 
Combined operations were first  considered by A. Salomaa, K. Salomaa and Yu~\cite{SSY07} in 2007. See also the 2012 paper on this topic by Brzozowski~\cite{Brz12} and the references in that paper.

It has been suggested in~\cite{BrYe11} by Brzozowski and Ye that syntactic complexity can be a useful measure of complexity.
It has its roots in the   \emph{Myhill congruence}~\cite{Myh57} $\approx_L$ defined by a language $L\subseteq \Sig^*$   as follows: For $x,y\in\Sig^*$,
\begin{equation*}
x \approx_L y \text{ if and only if } uxv\in L  \Leftrightarrow uyv\in L\text{ for all } u,v\in\Sig^*.
\end{equation*}
The 
\emph{syntactic semigroup}~\cite{Pin97} of $L$ is the quotient semigroup 
\mbox{$\Sig^+/\approx_L$}.
It is isomorphic to the semigroup of transformations of states by non-empty words in the minimal DFA of $L$.
This semigroup is called the \emph{transition semigroup} and  is often used to represent the syntactic semigroup.
\emph{Syntactic complexity} is the cardinality of the syntactic semigroup. Syntactic complexity may be able to distinguish between two regular languages with the same quotient complexity. For example, a language with three quotients may have syntactic complexity as low as 2 or as high as 27.

Atoms of regular languages were introduced in 2011~\cite{BrTa11}, and their quotient complexities were studied  in 2012~\cite{BrTa12}.
An \emph{atom}\footnote{ The definition of~\cite{BrTa11}, has been slightly modified in~\cite{BrTa12}.
The newer model, which admits up to $2^n$ atoms, is used here.} of a regular language $L$ with quotients $K_0,\ldots,K_{n-1}$ is a non-empty intersection of the form 
$\widetilde{K_0}\cap \cdots \cap \widetilde{K_{n-1}}$, 
where $\widetilde{K_i}$ is either $K_i$ or $\ol{K_i}$, and $\ol{K_i}=\Sig^*\setminus K_i$. Thus the number of atoms is bounded from above by $2^n$, and it was proved in~\cite{BrTa12} that this bound is tight.
Since every quotient of $L$ (including $L$ itself)  and every quotient of every atom of $L$ is a union of atoms, the atoms 
of $L$ are its basic building blocks. 
It was proved in~\cite{BrTa12} that the  quotient complexity of the atoms with 0 or $n$ complemented quotients is bounded from above by $2^n-1$, and that of any atom with $r$ complemented quotients, where $1\le r\le n-1$, by 
\begin{equation}
\label{eq:atombounds}
f(n,r)=1 + \sum_{k=1}^{r} \sum_{h=k+1}^{n-r+k} \binom{h}{n} \binom{k}{h}.
\end{equation}
These bounds are tight~\cite{BrTa12}.
When we say that a language has \emph{maximal quotient complexity of atoms} we mean that (a) it has all $2^n$ atoms, and (b) they all reach their maximal bounds, as stated above.

It was argued in~\cite{Brz12} that it is useful to consider several measures of complexity of regular languages, including syntactic complexity and atom complexity, along with the more traditional measures such as the state complexity of operations.
If one does consider several measures, the question arises whether these measures are related. 
There are only two such results. The first is the following proposition which restates for our purposes  the 2004 result of A. Salomaa, Wood, and Yu~\cite{SWY04}:
\begin{proposition}[Syntactic Semigroup and Reversal]
\label{prop:SWY}
Maximal syntactic complexity of a regular language implies maximal quotient complexity of its reverse.
\end{proposition}
In other words, if $L$ has syntactic complexity $n^n$, then the quotient complexity of $L^R$, the reverse of $L$, is necessarily $2^n$.

The converse of Proposition~\ref{prop:SWY} is false.
It was shown by Jir\'askov\'a and \v{S}ebej that the DFA\footnote{In the figure, if $n=2$, then $a$ transposes states 0 and 1, and $b$ is as shown. 
 For $n=3$, state 2 goes to itself under $b$.
For $n=4$, state 3 goes to itself under $a$.}  of Fig.~\ref{fig:rev}
with $n\ge 2$
 meets the upper bound for reversal~\cite{JiSe12}.
However, it is well known 
that  at least three inputs are required to generate all $n^n$ transformations when $n \ge 3$.
Thus the cardinality of the syntactic semigroup of the language of the DFA of Fig.~\ref{fig:rev}
is strictly smaller than $n^n$.

 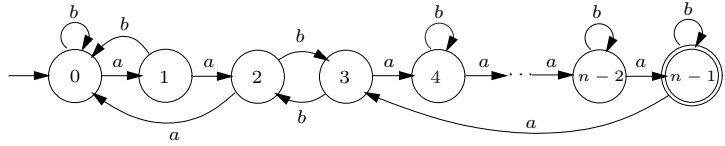
\begin{figure}[t]
 \begin{center}
 \setlength{\unitlength}{0.00043745in}
\begingroup\makeatletter\ifx\SetFigFont\undefined%
\gdef\SetFigFont#1#2#3#4#5{%
  \reset@font\fontsize{#1}{#2pt}%
  \fontfamily{#3}\fontseries{#4}\fontshape{#5}%
  \selectfont}%
\fi\endgroup%
{\renewcommand{\dashlinestretch}{30}
\begin{picture}(8460,1736)(0,-10)
\put(4998,748){\makebox(0,0)[lb]{\smash{{\SetFigFont{7}{8.4}{\rmdefault}{\mddefault}{\updefault}$4$}}}}
\put(5105.000,1286.333){\arc{333.333}{2.2143}{7.2105}}
\blacken\thicklines
\path(5236.377,1289.553)(5205.000,1153.000)(5302.310,1253.806)(5236.377,1289.553)
\thinlines
\put(6995.000,1278.333){\arc{333.333}{2.2143}{7.2105}}
\blacken\thicklines
\path(7126.377,1281.553)(7095.000,1145.000)(7192.310,1245.806)(7126.377,1281.553)
\thinlines
\put(8052.000,1331.333){\arc{333.333}{2.2143}{7.2105}}
\blacken\thicklines
\path(8183.377,1334.553)(8152.000,1198.000)(8249.310,1298.806)(8183.377,1334.553)
\thinlines
\put(1863,830){\ellipse{630}{630}}
\put(8092,833){\ellipse{720}{720}}
\put(793,828){\ellipse{630}{630}}
\put(8093,831){\ellipse{630}{630}}
\put(2956,820){\ellipse{630}{630}}
\put(3999,820){\ellipse{630}{630}}
\put(5094,828){\ellipse{630}{630}}
\put(6999,821){\ellipse{630}{630}}
\path(1107,830)(1557,830)
\blacken\thicklines
\path(1422.000,792.500)(1557.000,830.000)(1422.000,867.500)(1422.000,792.500)
\thinlines
\path(2172,823)(2622,823)
\blacken\thicklines
\path(2487.000,785.500)(2622.000,823.000)(2487.000,860.500)(2487.000,785.500)
\thinlines
\path(4325,830)(4775,830)
\blacken\thicklines
\path(4640.000,792.500)(4775.000,830.000)(4640.000,867.500)(4640.000,792.500)
\thinlines
\path(5420,830)(5870,830)
\blacken\thicklines
\path(5735.000,792.500)(5870.000,830.000)(5735.000,867.500)(5735.000,792.500)
\thinlines
\path(6208,838)(6658,838)
\blacken\thicklines
\path(6523.000,800.500)(6658.000,838.000)(6523.000,875.500)(6523.000,800.500)
\thinlines
\path(7318,823)(7768,823)
\blacken\thicklines
\path(7633.000,785.500)(7768.000,823.000)(7633.000,860.500)(7633.000,785.500)
\thinlines
\path(12,830)(455,830)
\blacken\thicklines
\path(320.000,792.500)(455.000,830.000)(320.000,867.500)(320.000,792.500)
\thinlines
\path(1670,1093)(1669,1094)(1667,1098)
	(1663,1103)(1657,1111)(1648,1122)
	(1638,1136)(1626,1151)(1612,1167)
	(1597,1184)(1581,1200)(1563,1217)
	(1545,1233)(1525,1247)(1503,1261)
	(1479,1272)(1454,1282)(1425,1290)
	(1395,1294)(1362,1295)(1329,1292)
	(1297,1285)(1267,1276)(1239,1264)
	(1213,1250)(1188,1235)(1165,1219)
	(1143,1202)(1122,1184)(1102,1166)
	(1083,1148)(1066,1131)(1051,1115)
	(1038,1101)(1027,1090)(1010,1070)
\blacken\thicklines
\path(1068.860,1197.149)(1010.000,1070.000)(1126.005,1148.575)(1068.860,1197.149)
\thinlines
\path(3207,1017)(3209,1018)(3213,1021)
	(3219,1026)(3229,1032)(3241,1040)
	(3256,1049)(3272,1059)(3291,1069)
	(3310,1079)(3332,1088)(3355,1096)
	(3381,1103)(3410,1109)(3442,1113)
	(3477,1115)(3512,1114)(3545,1111)
	(3575,1106)(3603,1099)(3628,1092)
	(3652,1083)(3674,1074)(3695,1065)
	(3714,1056)(3731,1047)(3745,1039)(3770,1025)
\blacken\thicklines
\path(3633.889,1058.242)(3770.000,1025.000)(3670.534,1123.680)(3633.889,1058.242)
\thinlines
\path(3755,613)(3753,612)(3749,609)
	(3743,604)(3733,598)(3721,590)
	(3706,581)(3690,571)(3671,561)
	(3652,551)(3630,542)(3607,534)
	(3581,527)(3552,521)(3520,517)
	(3485,515)(3450,516)(3417,519)
	(3387,524)(3359,531)(3334,538)
	(3310,547)(3288,556)(3267,565)
	(3248,574)(3231,583)(3217,591)(3192,605)
\blacken\thicklines
\path(3328.111,571.758)(3192.000,605.000)(3291.466,506.320)(3328.111,571.758)
\thinlines
\path(7820,598)(7819,598)(7818,597)
	(7815,595)(7811,592)(7804,588)
	(7796,582)(7785,576)(7772,567)
	(7756,557)(7737,546)(7716,534)
	(7692,520)(7666,505)(7638,489)
	(7607,472)(7574,454)(7540,436)
	(7504,418)(7466,399)(7426,380)
	(7385,362)(7342,344)(7298,326)
	(7252,308)(7204,291)(7154,274)
	(7102,258)(7048,243)(6991,228)
	(6931,215)(6869,202)(6803,190)
	(6733,180)(6660,170)(6584,162)
	(6504,156)(6420,151)(6334,149)
	(6245,148)(6164,149)(6082,153)
	(6002,157)(5923,164)(5845,171)
	(5769,180)(5696,190)(5624,201)
	(5555,212)(5487,225)(5421,238)
	(5357,252)(5295,266)(5234,282)
	(5174,297)(5115,313)(5058,330)
	(5002,346)(4947,363)(4893,381)
	(4840,398)(4788,416)(4738,433)
	(4689,450)(4642,467)(4597,484)
	(4554,500)(4513,515)(4475,530)
	(4440,544)(4407,556)(4378,568)
	(4352,578)(4329,588)(4309,596)
	(4293,602)(4279,608)(4269,612)
	(4261,615)(4250,620)
\blacken\thicklines
\path(4388.417,598.275)(4250.000,620.000)(4357.382,529.998)(4388.417,598.275)
\thinlines
\path(2688,629)(2687,628)(2685,626)
	(2681,622)(2674,616)(2665,608)
	(2653,597)(2639,585)(2621,569)
	(2601,553)(2579,534)(2554,515)
	(2528,494)(2500,474)(2470,453)
	(2439,432)(2406,412)(2372,392)
	(2336,373)(2299,355)(2259,338)
	(2216,323)(2171,309)(2123,296)
	(2071,285)(2017,277)(1960,271)
	(1901,269)(1846,270)(1791,273)
	(1738,279)(1687,287)(1638,297)
	(1592,309)(1548,322)(1506,336)
	(1465,351)(1426,367)(1389,384)
	(1353,402)(1318,420)(1284,438)
	(1251,457)(1220,476)(1190,494)
	(1162,512)(1136,529)(1113,544)
	(1092,558)(1074,571)(1059,581)
	(1046,590)(1037,597)(1023,607)
\blacken\thicklines
\path(1154.650,559.048)(1023.000,607.000)(1111.058,498.018)(1154.650,559.048)
\put(725,1520){\makebox(0,0)[lb]{\smash{{\SetFigFont{7}{8.4}{\familydefault}{\mddefault}{\updefault}$b$}}}}
\put(6364,935){\makebox(0,0)[lb]{\smash{{\SetFigFont{7}{8.4}{\familydefault}{\mddefault}{\updefault}$a$}}}}
\put(2307,935){\makebox(0,0)[lb]{\smash{{\SetFigFont{7}{8.4}{\familydefault}{\mddefault}{\updefault}$a$}}}}
\put(1235,935){\makebox(0,0)[lb]{\smash{{\SetFigFont{7}{8.4}{\familydefault}{\mddefault}{\updefault}$a$}}}}
\put(1325,1370){\makebox(0,0)[lb]{\smash{{\SetFigFont{7}{8.4}{\familydefault}{\mddefault}{\updefault}$b$}}}}
\put(3402,1235){\makebox(0,0)[lb]{\smash{{\SetFigFont{7}{8.4}{\familydefault}{\mddefault}{\updefault}$b$}}}}
\put(3425,275){\makebox(0,0)[lb]{\smash{{\SetFigFont{7}{8.4}{\familydefault}{\mddefault}{\updefault}$b$}}}}
\put(1903,73){\makebox(0,0)[lb]{\smash{{\SetFigFont{7}{8.4}{\familydefault}{\mddefault}{\updefault}$a$}}}}
\put(4445,935){\makebox(0,0)[lb]{\smash{{\SetFigFont{7}{8.4}{\familydefault}{\mddefault}{\updefault}$a$}}}}
\put(5915,792){\makebox(0,0)[lb]{\smash{{\SetFigFont{7}{8.4}{\familydefault}{\mddefault}{\updefault}$\cdots$}}}}
\put(5569,928){\makebox(0,0)[lb]{\smash{{\SetFigFont{7}{8.4}{\familydefault}{\mddefault}{\updefault}$a$}}}}
\put(7421,935){\makebox(0,0)[lb]{\smash{{\SetFigFont{7}{8.4}{\familydefault}{\mddefault}{\updefault}$a$}}}}
\put(6125,216){\makebox(0,0)[lb]{\smash{{\SetFigFont{7}{8.4}{\familydefault}{\mddefault}{\updefault}$a$}}}}
\put(6913,1520){\makebox(0,0)[lb]{\smash{{\SetFigFont{7}{8.4}{\familydefault}{\mddefault}{\updefault}$b$}}}}
\put(7993,1550){\makebox(0,0)[lb]{\smash{{\SetFigFont{7}{8.4}{\familydefault}{\mddefault}{\updefault}$b$}}}}
\put(5045,1520){\makebox(0,0)[lb]{\smash{{\SetFigFont{7}{8.4}{\familydefault}{\mddefault}{\updefault}$b$}}}}
\put(7833,767){\makebox(0,0)[lb]{\smash{{\SetFigFont{6}{7.2}{\familydefault}{\mddefault}{\updefault}$n-1$}}}}
\put(6745,767){\makebox(0,0)[lb]{\smash{{\SetFigFont{6}{7.2}{\familydefault}{\mddefault}{\updefault}$n-2$}}}}
\put(732,756){\makebox(0,0)[lb]{\smash{{\SetFigFont{7}{8.4}{\rmdefault}{\mddefault}{\updefault}$0$}}}}
\put(1789,740){\makebox(0,0)[lb]{\smash{{\SetFigFont{7}{8.4}{\rmdefault}{\mddefault}{\updefault}$1$}}}}
\put(2884,740){\makebox(0,0)[lb]{\smash{{\SetFigFont{7}{8.4}{\rmdefault}{\mddefault}{\updefault}$2$}}}}
\put(3926,740){\makebox(0,0)[lb]{\smash{{\SetFigFont{7}{8.4}{\rmdefault}{\mddefault}{\updefault}$3$}}}}
\thinlines
\put(792.000,1293.333){\arc{333.333}{2.2143}{7.2105}}
\blacken\thicklines
\path(923.377,1296.553)(892.000,1160.000)(989.310,1260.806)(923.377,1296.553)
\end{picture}
}
 \end{center}
 \caption{The DFA of a language meeting the bound $2^n$ for reversal.}  
 \label{fig:rev}
 \end{figure}
 
 The second result is the 2011 proposition of Brzozowski and Tamm~\cite{BrTa11,BrTa12}
\begin{proposition}[Number of Atoms and Reversal]
\label{prop:BrTa}
The number of atoms of a regular language is equal to the quotient complexity of its reverse.
\end{proposition}
 
The main result of this paper is the following theorem:
\begin{restatable}{theorem}{mainthm}{\bf (Syntactic Semigroup and Atoms)}
\label{thm:main}
Maximal  syntactic complexity of a regular language with $n$ quotients implies that the language has $2^n$ atoms and each atom has maximal quotient complexity.
\end{restatable}

The fact that the number of atoms of $L$ (quotient complexity of $L^R$) is $2^n$ does not imply that each atom has maximal quotient complexity. 
 For example, the language of Fig.~\ref{fig:rev}
for $n=4$ (respectively, $n=5,6,7$) has no atoms of quotient complexity larger than 25 (respectively, 99, 298,1053), but the maximal quotient complexity is 43 (respectively, 141, 501, 1548).

The converse of Theorem~\ref{thm:main} is not true. The language $L$ of the minimal DFA of Fig.~\ref{fig:wit} meets all the quotient complexity bounds for the 8 atoms, but its syntactic complexity is 24, while the maximum is 27.
There are also many ternary examples with higher numbers of states.
 \begin{figure}[t]
 \begin{center}
 \setlength{\unitlength}{0.00043745in}
\begingroup\makeatletter\ifx\SetFigFont\undefined%
\gdef\SetFigFont#1#2#3#4#5{%
  \reset@font\fontsize{#1}{#2pt}%
  \fontfamily{#3}\fontseries{#4}\fontshape{#5}%
  \selectfont}%
\fi\endgroup%
{\renewcommand{\dashlinestretch}{30}
\begin{picture}(3078,1405)(0,-10)
\put(1497,109){\makebox(0,0)[lb]{\smash{{\SetFigFont{7}{8.4}{\familydefault}{\mddefault}{\updefault}$a,b$}}}}
\put(2755,649){\ellipse{630}{630}}
\put(1685,647){\ellipse{630}{630}}
\put(2758,647){\ellipse{540}{540}}
\path(12,649)(282,649)
\blacken\path(162.000,619.000)(282.000,649.000)(162.000,679.000)(162.000,619.000)
\path(1992,649)(2442,649)
\blacken\path(2322.000,619.000)(2442.000,649.000)(2322.000,679.000)(2322.000,619.000)
\path(912,649)(1362,649)
\blacken\path(1242.000,619.000)(1362.000,649.000)(1242.000,679.000)(1242.000,619.000)
\path(1497,911)(1496,913)(1493,916)
	(1488,923)(1480,932)(1469,944)
	(1456,960)(1441,977)(1424,995)
	(1405,1014)(1385,1034)(1364,1052)
	(1341,1071)(1318,1087)(1292,1103)
	(1265,1116)(1236,1128)(1205,1137)
	(1172,1142)(1137,1144)(1103,1141)
	(1070,1134)(1039,1124)(1011,1111)
	(986,1096)(962,1080)(940,1062)
	(919,1043)(900,1023)(881,1003)
	(864,983)(849,964)(835,946)
	(824,931)(814,918)(799,896)
\blacken\path(841.814,1012.047)(799.000,896.000)(891.387,978.247)(841.814,1012.047)
\path(2562,387)(2561,386)(2559,384)
	(2555,381)(2548,376)(2539,368)
	(2528,359)(2513,347)(2496,333)
	(2475,318)(2453,300)(2428,282)
	(2400,262)(2372,242)(2341,222)
	(2309,201)(2276,180)(2241,160)
	(2205,141)(2168,122)(2129,104)
	(2089,87)(2046,72)(2002,57)
	(1955,44)(1905,33)(1853,24)
	(1798,17)(1742,13)(1684,12)
	(1626,14)(1570,19)(1516,27)
	(1464,37)(1415,49)(1368,63)
	(1324,78)(1282,94)(1243,112)
	(1205,130)(1168,150)(1133,170)
	(1099,191)(1067,212)(1036,233)
	(1006,254)(978,275)(952,295)
	(928,314)(906,331)(886,347)
	(869,361)(855,373)(844,383)
	(835,390)(822,402)
\blacken\thicklines
\path(946.634,337.987)(822.000,402.000)(895.763,282.877)(946.634,337.987)
\put(536,582){\makebox(0,0)[lb]{\smash{{\SetFigFont{7}{8.4}{\rmdefault}{\mddefault}{\updefault}$0$}}}}
\put(1624,582){\makebox(0,0)[lb]{\smash{{\SetFigFont{7}{8.4}{\rmdefault}{\mddefault}{\updefault}$1$}}}}
\put(2711,582){\makebox(0,0)[lb]{\smash{{\SetFigFont{7}{8.4}{\rmdefault}{\mddefault}{\updefault}$2$}}}}
\put(2120,747){\makebox(0,0)[lb]{\smash{{\SetFigFont{7}{8.4}{\familydefault}{\mddefault}{\updefault}$a$}}}}
\put(957,784){\makebox(0,0)[lb]{\smash{{\SetFigFont{7}{8.4}{\familydefault}{\mddefault}{\updefault}$a,b$}}}}
\put(1039,1219){\makebox(0,0)[lb]{\smash{{\SetFigFont{7}{8.4}{\familydefault}{\mddefault}{\updefault}$b$}}}}
\thinlines
\put(591,649){\ellipse{630}{630}}
\end{picture}
}
 \end{center}
 \caption{The DFA of a language with maximal quotient complexities of atoms, but not maximal syntactic complexity.}  
 \label{fig:wit}
 \end{figure}
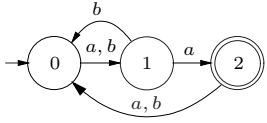

\medskip

The remainder of the paper is devoted to the proof of Theorem~\ref{thm:main}.

\section{Definitions}
\subsection{Automata and \'Atomata}
A~\emph{nondeterministic finite automaton (NFA)} is a quintuple 
$\cN=(Q, \Sig, \eta, I,F)$, where 
$Q$ is a finite, non-empty set of \emph{states}, 
$\Sig$ is a finite non-empty \emph{alphabet}, 
$\eta:Q\times \Sig\to 2^Q$ is the  \emph{transition function},
$I\subseteq  Q$ is the set of  \emph{initial states},
and $F\subseteq Q$ is the set of \emph{final states}.
For $a$ in $\Sig$, let $\eta_a : Q \rightarrow 2^Q$ be defined by $\eta_a(q) = \eta(q,a)$ for $q \in Q$. 
For $a \in \Sig$, $x \in \Sig^*$, and $w = xa$, define $\eta_w: Q \rightarrow 2^Q$ inductively by $\eta_w(q) = \eta_a(\eta_x(q))$. 

For any function $f : X \rightarrow Y$, we extend $f$ to subsets of the domain in the natural way by letting $f(S) = \bigcup_{s \in S} f(s)$ for $S \subseteq X$. Note $f(\emp) = \emp$ for all $f$.

The \emph{language accepted} by an NFA $\cN$ is 
$L(\cN)=\{w\in\Sig^*\mid \eta(I,w)\cap F\neq \emp\}$.
Two NFAs are \emph{equivalent} if they accept the same language. 
The \emph{left language} of a state $q$  is
$L_{I,q}=\{w\in\Sig^* \mid q\in \eta(I,w)\}$.
The \emph{right language} of a state $q$  is
$L_{q,F}(\cN)=\{w\in\Sig^* \mid \eta(q,w)\cap F\neq\emp\}$.
The \emph{right language} of a set $S$ of states of $\cN$ is
$L_{S,F}(\cN)=\bigcup_{q\in S} L_{q,F}(\cN)$; so
$L(\cN)=L_{I,F}(\cN)$.
A state is \emph{unreachable} if its left language is empty and \emph{reachable} otherwise.
A set $S$ of states is \emph{strongly connected} if for all $p,q \in S$, there exists $w \in \Sig^*$ such that $\eta(p,w) = q$.
An NFA is \emph{minimal} if it has the minimal number of states among all
the equivalent NFAs.

A \emph{deterministic finite automaton (DFA)} is a quintuple 
$\cD=(Q, \Sig, \delta, q_0,F)$, where
$Q$, $\Sig$, and $F$ are as in an NFA, 
$\delta:Q\times \Sig\to Q$ is the transition function, 
and $q_0$ is the initial state. 
It is clear that a DFA is a special type of NFA, so the definitions stated above for NFAs also apply to DFAs.
It is well-known that for every regular language $L$, there exists a unique (up to isomorphism) minimal DFA. Furthermore, there is a one-to-one correspondence between the states of the minimal DFA and the quotients of $L$.

For an NFA $\cN$ (or DFA $\cD$), let $\cN^\rev$ (or $\cD^\rev$) denote the result of performing the \emph{reversal} operation which interchanges the final and initial states, and reverses all the transitions. 
Let $\eta^{\rev}$ (or $\delta^\rev$) denote the transition function of $\cN^{\rev}$ (or $\cD^\rev$).

Let $\cN^\deter$ denote the result of performing the \emph{determinization} operation, which is the well-known subset construction. Unreachable subsets are not included in the determinization, but the empty state, if present, is included.
Let $\eta^{\deter}$ denote the transition function of $\cN^{\deter}$. 

For $S \subseteq Q$, let $A_S$ denote the following intersection of uncomplemented and complemented quotients:
\begin{equation}
\label{eqn:atom}
A_S = \left(\bigcap_{i \in S} K_i\right) \cap \left(\bigcap_{j \in Q \setminus S} \ol{K_j}\right).
\end{equation}
An \emph{atom}~\cite{BrTa11,BrTa12} of $L$ is such an intersection $A_S$, provided it is not empty.
If the intersection with all quotients complemented is non-empty, then it constitutes the \emph{negative} atom;  
all the other atoms are \emph{positive}. 
Let $A=\{A_0,\ldots,A_{m-1}\}$ be the set of atoms of $L$,  and let the number of positive atoms be $p$.
The only atom containing $\eps$ is the one in which all the quotients containing $\eps$ are uncomplemented and all the remaining quotients  are complemented. 
This atom is called \emph{final}, and is $A_{p-1}$ by convention.
The negative atom can never be final if $L$ is non-empty, since there must be at least one final quotient in its intersection.
Atoms containing $L$, rather than $\ol{L}$ in their intersection  are called \emph{initial}.

We use the one-to-one correspondence between atoms $A_i$ and \emph{atom symbols}
$\bA_i$. Let $\bA=\{\bA_0,\ldots,\bA_{m-1}\}$ be the set of atom symbols.

\begin{definition}
\label{def:atomaton}
The \emph{\'atomaton} of $L$
 is the NFA $\cA=(\bA,\Sig,\eta,\bA_I, \{\bA_{p-1}\}),$
 where $\bA$ is the set of atom symbols,
 $\bA_I$ corresponds to the set of initial atoms, 
 $\bA_{p-1}$ corresponds to the final atom,
 and $\bA_j \in \eta(\bA_i, a)$ if and only if 
$aA_j \subseteq A_i$, for all $\bA_i,\bA_j \in \bA$ and $a\in\Sig$.
\end{definition}
In the \'atomaton, the right language of any state $\bA_i$ is the atom $A_i$~\cite{BrTa11}.
Also, all the positive atoms are reachable, but the negative atom is not.

It was shown in~\cite{BrTa11,BrTa12}  that $\cA^\rev$ is a minimal DFA 
that accepts $L^R$, and  that $\cA^\rev$
is isomorphic to $\cD^{\rev\deter}$.
The following makes this isomorphism precise~\cite{BrTa12}:

\begin{proposition}[\'Atomaton Isomorphism]
\label{prop:isomorphism}
Let $L$ be a regular language and let $K$ be its  set of quotients.
Let $\varphi: A \to 2^K$ be the mapping assigning to state 
$\bA_j$, corresponding to atom
$A_j= \left(\bigcap_{i \in S} K_i\right) \cap \left(\bigcap_{j \in Q \setminus S} \ol{K_j}\right)$ of
 $\cA^\rev$, the set 
$S$.
Then $\varphi$ is a DFA isomorphism between $\cA^\rev$ and 
$\cD^{\rev\deter}$. 
\end{proposition}

\begin{corollary}
\label{cor:isomorphism}
The mapping $\varphi$ is an NFA isomorphism between 
$\cA$ and $\cD^{\rev\deter\rev}$.
\end{corollary}

\subsection{Transformations}
A {\em transformation} of a set $Q$ is a mapping of $Q$ into itself. We consider only transformations $t$ of a finite set $Q$.
For a transformation $t$ of $Q$ and a subset $S$ of $Q$, let $t^{-1}(S) = \{q \in Q \mid \text{there exists } i \in S \text{ such that } t(q) = i\}$. We say $t^{-1}(S)$ is the \emph{preimage} of $S$ under $t$: the maximal set of elements of $Q$ that is mapped onto $S$ by $t$. When discussing preimages of singletons such as $t^{-1}(\{i\})$, we drop the braces and write $t^{-1}(i)$. If $P \subseteq Q$ is in the set $\preim t = \{ P \mid \text{there exists } S \subseteq Q \text{ such that } P = t^{-1}(S)\}$, then we say $P$ is \emph{a preimage} of $t$ (as opposed to calling it \emph{the preimage} of some $S$). The set $\preim t$ is the set of all preimages of $t$.

The \emph{image} of $t$ is $\im t = \{q \in Q \mid \text{there exists } p \in Q \text{ such that } t(p) = q\}$; this is the subset of $Q$ that $t$ maps onto. The \emph{coimage} of $t$ is $\coim t = Q \setminus \im t$; this is  the set of elements of $Q$ that are not mapped onto $\im t$. For $P \subseteq Q$, the set $t(P)$ obtained by applying $t$ to each element of $P$ is called the \emph{image of $P$ under $t$}.


A transformation $t$ is a \emph{cycle} of length $k$, where $k \ge 2$, if there exist pairwise different elements $i_1,\ldots,i_k$ such that
$t(i_1)=i_2,  t(i_2)=i_3,\ldots, t(i_{k-1})=i_k$, and $t(i_k)=i_1$, and the remaining elements are mapped to themselves.
A~cycle is denoted by $(i_1,i_2,\ldots,i_k)$.
For $i<j$, a \emph{transposition} is the cycle $(i,j)$.
A~\emph{singular} transformation, denoted by $(i\rightarrow j)$, has $t(i)=j$ and $t(h)=h$ for all $h\neq i$.
A~\emph{constant} transformation,  denoted by $(Q \rightarrow j)$, has $t(i)=j$ for all $i$.
If $s$ and $t$ are transformations, the \emph{composition} $s \circ t$ is defined by $s \circ t(i) = s(t(i))$.

%
%
%
%

\section{Proof of the Main Result}

To establish Theorem~\ref{thm:main}, we need several intermediate results. 
In the sequel we represent the states of the \'atomaton $\cA$ of a regular language $L$ by sets of quotients of $L$, that is, by sets of states of the minimal DFA $\cD$ recognizing $L$, as allowed by Proposition~\ref{prop:isomorphism}.
Since the states of $\cA$ are sets of states and $\cA$ is an NFA, the outputs of $\cA$'s transition function are \emph{sets of sets} of states. To reduce confusion, we refer to these as \emph{collections of sets} of states.

In some case, the collections of sets that arise as outputs of $\cA$'s transition function can be described as ``intervals''. If $U$ and $V$ are sets, the \emph{interval $[V,U]$ between $V\!$ and $U$} is the collection of all subsets of $U$ that contain $V$. Note that if $V$ is not a subset of $U$, this interval is empty. 

\subsection{Transition Function of the \'Atomaton}
\begin{lemma}
\label{lem:atrans}
Let $L\subseteq\Sig^*$ be a regular language  with quotient complexity $n$ and syntactic complexity $n^n$. Let $\cD$ be the minimal DFA for $L$ with state set $Q$ and transition function $\delta$. Let $\cA$ be the \'atomaton of $L$ with transition function $\eta$. 

\be
\item
\label{lem:atrans:letter}
Let $S \subseteq Q$ and $a \in \Sig$. Then the transition function $\eta$ of $\cA$ satisfies:
\[ \eta_a(S) = 
\begin{cases}
[\delta_a(S),\delta_a(S) \cup \coim\delta_a],&
\text{ if $S \in \preim \delta_a$; }\\
\emp,&\text{ otherwise. }
\end{cases} \]

\item
\label{lem:atrans:preim}
Let $U,V \subseteq Q$ and let $a \in \Sig$. If every set in the interval $[V,U]$ is a preimage of $\delta_a$, then the transition function $\eta$ of $\cA$ satisfies:
\[ \eta_a([V,U]) = [\delta_a(V),\delta_a(U) \cup \coim\delta_a]. \]

\item
\label{lem:atrans:perm}
Let $U,V \subseteq Q$ and let $w \in \Sig^*$. If $\delta_w$ is a permutation, then the transition function $\eta$ of $\cA$ satisfies:
\[ \eta_w([V,U]) = [\delta_w(V),\delta_w(U)]. \]
\ee
\end{lemma}


\begin{proof}
(\ref{lem:atrans:letter}):
In $\cD^{\rev}$, the letter $a$ induces the function $\delta_a^\rev : Q \rightarrow 2^{Q}$ which maps each state $i$ to its preimage $\delta_a^{-1}(i)$. Furthermore, by Proposition~\ref{prop:SWY}, every subset of $Q$ is reachable in $\cD^{\rev\deter}$ since $L$ has maximal syntactic complexity. So the set of states of $\cD^{\rev\deter}$ is $2^Q$, and the empty set of states of $\cD$ is a state of $\cD^{\rev\deter}$. Thus in $\cD^{\rev\deter}$, the letter $a$ induces the function $\delta_a^{\rev\deter} : 2^{Q} \rightarrow 2^{Q}$, defined as follows:
\begin{equation}
\label{eqn:deltaRD}
\delta_a^{\rev\deter}(S) =
\displaystyle{\bigcup_{i \in S} \delta_a^{-1}(i)} = \delta_a^{-1}(S).
\end{equation}
In $\cD^{\rev\deter\rev}$, $a$ induces the function $\delta^{\rev\deter\rev}_a : 2^{Q} \rightarrow 2^{2^Q}$. This function maps a subset $S$ of $Q$ to its preimage under $\delta_a^{\rev\deter}$, that is, to the collection of sets each of which maps to $S$ under $\delta_a^{\rev\deter}$. Since $\cD^{\rev\deter\rev}$ is isomorphic to $\cA$, $\delta^{\rev\deter\rev}$ is equivalent to $\eta$. We now show this function satisfies the statement from the lemma.

Notice that if $S \not\in \preim\delta_a$, then $S$ cannot be an output of $\delta_a^{\rev\deter}$. It follows that $\delta_a^{\rev\deter\rev}(S) = \emp$, since the collection of sets that map to $S$ under $\delta_a^{\rev\deter}$ is empty.

Conversely, suppose $S \in \preim\delta_a$. Then clearly $S$ is the preimage of $\delta_a(S)$ under $\delta_a$. It follows by Equation (\ref{eqn:deltaRD}) that $\delta^{\rev\deter}_a(\delta_a(S)) = S$, and thus $\delta_a(S)$ is in the collection of sets produced by $\delta_a^{\rev\deter\rev}(S) = \eta_a(S)$.

Consider which other sets map to $S$ under $\delta_a^{\rev\deter}$. 
Notice no strict subset of $\delta_a(S)$ maps to $S$; if $S$ is the preimage of $\delta_a(T) \subset \delta_a(S)$ under $\delta_a$ this means $\delta_a^{-1}(\delta_a(T)) = S$, and applying $\delta_a$ to both sides gives $\delta_a(T) = \delta_a(S)$.
A strict superset of $\delta_a(S)$, say $\delta_a(S) \cup T$, maps to $S$ only if $\delta_a^{-1}(T) \subseteq S$, since we have
\[ \delta_a^{\rev\deter}(\delta_a(S) \cup T) 
= \delta_a^{-1}(\delta_a(S) \cup T) 
= \delta_a^{-1}(\delta_a(S)) \cup \delta_a^{-1}(T)
= S \cup \delta_a^{-1}(T). \]
Suppose $\delta_a^{-1}(T)$ is non-empty. Since $\delta_a^{-1}(T) \subseteq S$, we have $T \subseteq \delta_a(S)$. Thus if $\delta_a(S) \cup T$ is a strict superset of $\delta_a(S)$, then $\delta_a^{-1}(T)$ must be empty. Therefore $T$ must be a subset of $\coim\delta_a$, since $\coim\delta_a$ contains all the elements of $Q$ with empty preimages under $\delta_a$.

In fact, for all $T \subseteq \coim\delta_a$ we have $\delta_a^{\rev\deter}(\delta_a(S) \cup T) = S$. This means that the collection of sets produced by $\delta_a^{\rev\deter\rev}(S)$ (and thus $\eta_a(S)$) is the set of all supersets of $\delta_a(S)$ which are subsets of $\delta_a(S) \cup \coim\delta_a$. Thus, as required, we have: 
\[ \eta_a(S) = [\delta_a(S),\delta_a(S) \cup \coim\delta_a]. \]

(\ref{lem:atrans:preim}):
We proceed by induction on the number of sets in the interval. If there are no sets, that is,  if $[V,U] = \emp$, then $\eta_a([V,U]) = \eta_a(\emp) = \emp$ as required. If there is only one set, say $[V,U] = [S,S] = \{S\}$, then the proof of the previous part shows the statement is true.

Suppose that  the statement holds if $|[V,U]| < k$. We must show it also holds if $|[V,U]| = k$. If $V \supset U$ then $[V,U] = \emp$, and if $V = U$ then $|[V,U]| = 1$. These are the base cases, so we can assume that $V \subset U$.

If $V \subset U$, then we have some $u \in U$ such that $u \not\in V$. Notice that we can write $[V,U]$ as $[V \cup \{u\},U] \cup [V,U \setminus \{u\}]$. Also, we have $\eta_a([V \cup \{u\},U] \cup [V,U \setminus \{u\}]) = \eta_a([V \cup \{u\},U]) \cup \eta_a([V,U \setminus \{u\}])$. It follows that:
\[ \eta_a([V,U]) = \eta_a([V \cup \{u\},U]) \cup \eta_a([V,U \setminus \{u\}]). \]
These two intervals have strictly fewer sets than $[V,U]$; so by the induction hypothesis we have:
\[ \eta_a([V \cup \{u\},U]) = [\delta_a(V \cup \{u\}),\delta_a(U) \cup \coim\delta_a] \text{, and} \]
\[ \eta_a([V,U \setminus \{u\}]) = [\delta_a(V),\delta_a(U \setminus \{u\}) \cup \coim\delta_a]. \]
Notice that $U$ and $U \setminus \{u\}$ are both in $[V,U]$, and thus are preimages of $\delta_a$. Since preimages are maximal, distinct preimages  map to distinct sets under $\delta_a$. Thus $\delta_a(U \setminus \{u\}) \ne \delta_a(U)$. It follows that $\delta_a(u) \not\in \delta_a(U \setminus \{u\})$, since otherwise the two sets would be equal. Furthermore, $\delta_a(u)$ is the only element which is present in $\delta_a(U)$ but not present in $\delta_a(U \setminus \{u\})$. Thus $\delta_a(U \setminus \{u\}) = \delta_a(U) \setminus \{\delta_a(u)\}$. It follows that:
\[ \eta_a([V,U \setminus \{u\}]) = [\delta_a(V),\delta_a(U) \setminus \{\delta_a(u)\} \cup \coim\delta_a]. \]
Furthermore, noting that $\delta_a(V \cup \{u\}) = \delta_a(V) \cup \{\delta_a(u)\}$, we have:
\[ \eta_a([V \cup \{u\},U]) = [\delta_a(V) \cup \{\delta_a(u)\},\delta_a(U) \cup \coim\delta_a]. \]
Thus, as required, the union of these two intervals is:
\[ \eta_a([V,U \setminus \{u\}])  \cup \eta_a([V \cup \{u\},U])
= [\delta_a(V),\delta_a(U) \cup \coim\delta_a] .\]

(\ref{lem:atrans:perm}):
We proceed by induction on the length of $w$.
Every subset of $Q$ is a preimage of $\delta_w$, since $\delta_w$ is a permutation. Also, $\coim\delta_w = \emp$. Thus the base case (where $w$ is a single letter) is covered by the proof of the previous part.

Now suppose $w = a_1a_2\dotsb a_k$ and the lemma holds for words of length less than $k$. Let $w' = a_1a_2\dotsb a_{k-1}$. By the inductive hypothesis, we have
\[ \eta_{w'}([V,U]) = [\delta_{w'}(V),\delta_{w'}(U)]. \]
Notice that $\delta_w = \delta_{a_k} \circ \delta_{w'}$, and similarly $\eta_w = \eta_{a_k} \circ \eta_{w'}$. Furthermore, $\delta_{a_k}$ must be a permutation (or else $\delta_w$ would not be a permutation). Thus by Part \ref{lem:atrans:preim} of this lemma, we have:
\begin{align*}
\eta_w([V,U]) 
&= \eta_{a_k}(\eta_{w'}([V,U])) 
= \eta_{a_k}\left([\delta_{w'}(V),\delta_{w'}(U)]\right) \\
&= [\delta_{a_k}(\delta_{w'}(V)),\delta_{a_k}(\delta_{w'}(U))]
= [\delta_{w}(V),\delta_{w}(U)].
\end{align*}
This proves that the statement holds for $k$ and thus for all natural numbers.
\qed
\end{proof}

\begin{example}
Consider the DFA $\cD$ with $Q = \{0,1,2\}$, $\Sig = \{a,b,c,d\}$, $q_0 = 0$, $F = \{2\}$, and transition function $\delta$ defined by $\delta_a = (0,1)$, $\delta_b = (1,2)$, $\delta_c = {(2 \rightarrow 0)}$, and $\delta_d = {(Q \rightarrow 1)}$. The language $L = L(\cD)$ has syntactic complexity $n^n=3^3=27$. 

The transition functions of $\cD$, $\cD^{\rev}$, $\cD^{\rev\deter}$ and $\cA = \cD^{\rev\deter\rev}$ are shown in Tables \ref{ex:firsttable} to \ref{ex:lasttable}. For conciseness, we represent  sets like $\{0,1,2\}$ and $\{0,2\}$ by $012$ and $02$, respectively, and  collections of sets like $\{\{0\},\{0,1\},\{0,2\},\{0,1,2\}\}$, by $0,01,02,012$. We use $\Phi$ to denote the ``empty-set state'' that arises when performing determinization of an NFA $\cN$ (that is, the state in $\cN^{\deter}$ which corresponds to the empty subset of states of $\cN$) and $\emp$ to denote the actual empty set. The arrows in the leftmost column of each table denote initial states $(\rightarrow)$ and final states $(\leftarrow)$.
\renewcommand{\arraystretch}{1.1}
\begin{table}
\begin{minipage}[c]{0.5\linewidth}
\centering
\caption{DFA $\cD$.}
\begin{tabular}{|c|c||c|c|c|c|}
\hline
 & \ \ $\delta$ \ \
    &\ \ $a$\ \ &\ \ $b$\ \ &\ \ $c$\ \ &\ \ $d$\ \ \\\hline\hline
$\rightarrow$ & 0   &1  &0  &0  &1  \\\hline
& 1   &0  &2  &1  &1  \\\hline
$\leftarrow$ & 2   &2  &1  &0  &1  \\\hline
\end{tabular}
\label{ex:firsttable}
\end{minipage}
\begin{minipage}[c]{0.5\linewidth}
\centering
\caption{NFA $\cD^\rev$.}
\begin{tabular}{|c|c||c|c|c|c|}
\hline
& \ \ $\delta^{\rev}$ \ \
    &\ \ $a$ \  \  &\ \ $b$ \  \ &\ \ $c$  \ \ &\ \ $d$ \  \ \\\hline\hline
$\leftarrow$ & 0   &1      &0      &02     &$\emp$ \\\hline
& 1   &0      &2      &1      &012    \\\hline
$\rightarrow$ & 2   &2      &1      &$\emp$ &$\emp$ \\\hline
\end{tabular} 
\end{minipage}
\end{table}

\begin{table}
\begin{minipage}[c]{0.40\linewidth}
\caption{$\cD^{\rev\deter}$.}
\centering
\begin{tabular}{|c|c||c|c|c|c|}
\hline
& \  $\delta^{\rev\deter}$ \ 
        &\ \ $a$  \ \ &\ \ $b$ \ \  &\ \ $c$ \ \  &\ \ $d$  \ \ \\\hline\hline
& $\Phi$  &$\Phi$ &$\Phi$ &$\Phi$ &$\Phi$ \\\hline
$\leftarrow$ & 0       &1      &0      &02     &$\Phi$ \\\hline
& 1       &0      &2      &1      &012    \\\hline
$\rightarrow$ & 2       &2      &1      &$\Phi$ &$\Phi$ \\\hline
$\leftarrow$ & 01      &01     &02     &012    &012    \\\hline
$\leftarrow$ & 02      &12     &01     &02     &$\Phi$ \\\hline
& 12      &02     &12     &1      &012    \\\hline
$\leftarrow$ & 012     &012    &012    &012    &012    \\\hline
\end{tabular} 
\end{minipage}
\hspace{.05cm}
\begin{minipage}[c]{0.60\linewidth}
\centering
\caption{$\cA = \cD^{\rev\deter\rev}$.}
\begin{tabular}{|c|c||c|c|c|c|}
\hline
&  $\eta = \delta^{\rev\deter\rev}$ 
        &$a$                &$b$                &$c$                &$d$                \\\hline\hline
& $\Phi$  &$\Phi$             &$\Phi$             &$\Phi$,2           &$\Phi$,0,2,02      \\\hline
$\rightarrow$ & 0       &1                  &0                  &$\emp$             &$\emp$             \\\hline
& 1       &0                  &2                  &1,12               &$\emp$             \\\hline
$\leftarrow$ & 2       &2                  &1                  &$\emp$             &$\emp$             \\\hline
$\rightarrow$ & 01      &01                 &02                 &$\emp$             &$\emp$             \\\hline
$\rightarrow$ & 02      &12                 &01                 &0,02               &$\emp$             \\\hline
& 12      &02                 &12                 &$\emp$             &$\emp$             \\\hline
$\rightarrow$ & 012     &\ 012    \            &\ 012    \            &\ 01,012  \           &\ 1,01,12,012 \  \\\hline
\end{tabular}
\label{ex:lasttable}
\end{minipage}
\end{table}

One can check that the definition of the transition function $\eta = \delta^{\rev\deter\rev}$ of the \'atomaton matches that of Part \ref{lem:atrans:letter} of the lemma. For example, we have $\eta_d(\{0,1,2\}) = \{\{1\},\{0,1\},\{1,2\},\{0,1,2\}\} = [\{1\},\{0,1,2\}]$. The lower bound of this interval is $\{1\}$ = $\delta_d(\{0,1,2\})$. Since $\coim\delta_d = \{0,2\}$, the upper bound of this interval is $\delta_d(\{0,1,2\}) \cup \coim\delta_d = \{1\} \cup \{0,2\} = \{0,1,2\}$.

Notice that $\{0,1,2\}$ is a preimage of $\delta_d$ (in particular, $\delta^{-1}_d(1) = \{0,1,2\}$) so $\eta_d(\{0,1,2\})$ is not the empty set. The only other preimage of $\delta_d$ is $\Phi$, and we have $\eta_d(\{\Phi\}) = [\Phi,\{0,2\}]$ as required. For all other subsets $S$ of $\{0,1,2\}$, we see that $S$ is not a preimage of $\delta_d$ and $\eta_d(S) = \emp$ as required.
\qedb
\end{example}

\subsection{Strong Connectedness and Reachability}
To show that each atom has maximal quotient complexity if the associated language has maximal syntactic complexity, we follow the approach of~\cite{BrTa12}. 
Let $L\subseteq \Sig^*$ be a regular language, let $\cD$ be the quotient DFA for $L$ with state set $Q$, and let $\cA$ be the \'atomaton of $L$. For $S \subseteq Q$, we derive $\cA^{\deter}_S$ (the minimal DFA of the atom $A_S$) by making  $S$ the starting state of $\cA$, and then determinizing. 
The initial state of $\cA^{\deter}_S$ is $\{S\}$, or equivalently the interval $[S,S]$. To prove the quotient complexity of $A_S$ is maximal, we use our results on the transition function of $\cA$ (Lemma \ref{lem:atrans}) to count the number of intervals that are reachable from $[S,S]$ in $\cA^{\deter}_S$. If the number of reachable intervals meets the quotient complexity bound for the atom $A_S$, it follows $A_S$ has maximal quotient complexity.

First we prove the following lemma:
\begin{lemma}
\label{lem:denes}
Let $L\subseteq \Sig^*$ be a regular language  with quotient complexity $n$ and syntactic complexity $n^n$. Let $\cD$ be the minimal DFA of $L$ with transition function $\delta$ and state set $Q$. Then there exists $a \in \Sig$ and $w \in \Sig^*$ such that $\delta_a = \alpha \circ \delta_w$, where $\alpha$ is a singular transformation and $\delta_w$ is a permutation.
\end{lemma}

\begin{proof}
Let $T = \{\delta_a \mid a \in \Sig\}$. Since $L$ has syntactic complexity $n^n$, the set $T$ generates all transformations of $Q$. We claim there exists $\delta_a \in T$ such that $|\im\delta_a| = n-1$.

To see this, observe that if $s$ and $t$ are transformations with $|\im s| = k$ and $|\im t| = \ell$, then $|\im(s \circ t)| \le \min\{k,\ell\}$. 
Now suppose for a contradiction that for all $\delta_a \in T$, we have $|\im\delta_a| = n$ or $|\im\delta_a| = n-2$. Since $T$ generates all transformations of $Q$, there exists $w \in \Sig^*$ such that $|\im\delta_w| = n-1$. Clearly $w$ cannot contain any letter $b \in \Sig$ such that $|\im\delta_b| \le n-2$, or else we would have $|\im\delta_w| \le |\im\delta_b| < n-1$. It follows $w$ only contains letters $b$ such that $|\im\delta_b| = n$. Thus $\delta_w$ is a permutation, since it is a composition of permutations. But this implies $|\im\delta_w| = n$, which is a contradiction.

Thus there exists $a \in \Sig$ such that $|\im\delta_a| = n-1$. 
Suppose $\im\delta_a = \{q_1, q_2, \dotsc, q_{n-1}\}$ and $\coim\delta_a = \{q_n\}$. Since $|\im\delta_a| = n - 1$, there exists a subset $P = \{p_1, p_2, \dotsc, p_{n-1}\}$ of $Q$ such that $\delta_a(p_i) \ne \delta_a(p_j)$ for all $i,j$. Suppose without loss of generality that $\delta_a(p_i) = q_i$. 

In $Q \setminus P$ there is precisely one state, say $p_n$. Since $p_n \not\in P$, we have $\delta_a(p_n) = \delta_a(p_j) = q_j$ for exactly one $p_j \in P$. 

Recall that for all transformations $t$ of $Q$, there exists $w \in \Sig^*$ that induces $t$. Pick $w$ such that $\delta_w: Q \rightarrow Q$ satisfies $\delta_w(p_i) = q_i$ for all $p_i$.
Notice that $\delta_w$ is a permutation. Now let $\alpha: Q \rightarrow Q$ be the singular transformation ${(q_n \rightarrow q_j)}$. Then $\alpha(\delta_w(p_i)) = \alpha(q_i) = q_i$ for all $p_i \in P$, and $\alpha(\delta_w(p_n)) = \alpha(q_n) = q_j$. Thus $\alpha \circ \delta_w = \delta_a$ as required. \qed
\end{proof}

Now we can prove the main result of this section. 
We assign a \emph{type} to all non-empty intervals as follows: the type of $[V,U]$ is the ordered pair $(|V|,|U|)$. For example, $[\{1,2\},\{1,2,3,4\}]$ has type $(2,4)$ and $[\emp,\emp]$ has type $(0,0)$. The interval $[\{1,2\},\{3,4\}]$ is empty and thus has no type.

\begin{lemma}
\label{lem:reach}
Suppose that $L$ has quotient complexity $n$ and syntactic complexity $n^n$. Consider $S \subseteq Q$ and $\cA^{\deter}_S$, the minimal DFA of the atom $A_S$.
\be
\item
\label{lem:reach1}
All states of $\cA^{\deter}_S$ which are intervals of the same type are strongly connected.
\item
\label{lem:reach2}
From a state in $\cA^{\deter}_S$ which is a interval of type $(v,u)$, if $v \ge 2$ we can reach a state which is a interval of type $(v-1,u)$ and if $u \le n-2$ we can reach a state which is a interval of type $(v,u+1)$.
\ee
\end{lemma}

\begin{proof}
(\ref{lem:reach1}): 
Since $L$ has syntactic complexity $n^n$, every permutation of $Q$ can be induced by a word in $\Sig^*$.
Let $[V_1,U_1]$ and $[V_2,U_2]$ be states of $\cA^{\deter}_S$ of the same type. We can assume that $V_1 \subseteq U_1$ and $V_2 \subseteq U_2$. Now consider $w \in \Sig^*$ and suppose $\delta_w : Q \rightarrow Q$ is a permutation that sends $V_1$ to $V_2$ and $U_1$ to $U_2$; such a permutation exists if $V_1 \subseteq U_1$ and $V_2 \subseteq U_2$. By Part \ref{lem:atrans:perm} of Lemma \ref{lem:atrans}, we have:
\[ \eta_w([V_1,U_1]) = [\delta_w(V_1),\delta_w(U_1)] = [V_2,U_2] .\]
Thus any two intervals of the same type in $\cA^{\deter}_S$ are connected by a word in $\Sig^*$.

(\ref{lem:reach2}): By Lemma \ref{lem:denes}, there exists a single letter $a \in \Sig$ and a word $w \in \Sig^*$ such that $\delta_a$ induces a transformation $\alpha \circ \delta_w$, where $\alpha$ is a singular transformation and $\delta_w$ is a permutation. Suppose $\alpha = {(k \rightarrow \ell)}$ for $k,\ell \in Q$. 

Note that a subset $S$ of $Q$ is a preimage of $\alpha$ only if $\{k,\ell\} \subseteq S$ or $\{k,\ell\} \cap S = \emp$. 
Since $\delta_a = \alpha \circ \delta_w$, it follows that $S$ is a preimage of $\delta_a$ only if $\{\delta_w^{-1}(k),\delta_w^{-1}(\ell)\} \subseteq S$ or $\{\delta_w^{-1}(k),\delta_w^{-1}(\ell)\} \cap S = \emp$. Also note that since $\delta_a = \alpha \circ \delta_w$ and $\coim\delta_w = \emp$, we have $\coim\delta_a = \coim\alpha = \{k\}$.

Let $[V,U]$ be a interval of type $(v,u)$ with $v \ge 2$. By Part \ref{lem:reach1} of this lemma, from $[V,U]$ we can reach a interval $[V',U']$ of type $(v,u)$ such that $\{k,\ell\} \subseteq V'$, and thus $k$ and $\ell$ are in every set of $[V',U']$. Since $L$ has syntactic complexity $n^n$, there exists  $x \in \Sig^*$ such that $\delta_x = \delta_w^{-1}$. By Part \ref{lem:atrans:perm} of Lemma \ref{lem:atrans}, we can apply $\eta_x$ to $[V',U']$ to obtain $[\delta_w^{-1}(V'),\delta_w^{-1}(U')]$. Every set in this interval is a preimage of $\delta_a$ since every set contains both $\delta_w^{-1}(k)$ and $\delta_w^{-1}(\ell)$.

By  Lemma \ref{lem:atrans}, Part \ref{lem:atrans:preim},  $\eta_a([\delta_w^{-1}(V'),\delta_w^{-1}(U')])$ is $[\alpha(V'),\alpha(U') \cup \{k\}]$ (since $\delta_a = \alpha \circ \delta_w$, $\delta_w$ cancels  its inverse).
Since $\{k,\ell\} \subseteq V' \subseteq U'$, we have $\alpha(V') = V' \setminus \{k\}$ and $\alpha(U') \cup \{k\} = U' \setminus \{k\} \cup \{k\} = U'$. Thus the resulting interval is $[V' \setminus \{k\},U']$, which has type $(v-1,u)$ as required.

In a similar fashion, suppose we have an interval $[V,U]$ of type $(v,u)$ such that $u \le n-2$. We can reach $[V',U']$ such that $\{k,\ell\} \cap U' = \emp$. We can then apply $\eta_x$ to get $[\delta_w^{-1}(V'),\delta_w^{-1}(U')]$. As before, each set in this interval is a preimage of $\delta_a$ since for all sets $S$ in the interval we have $\{\delta_w^{-1}(k),\delta_w^{-1}(\ell)\} \cap S = \emp$. Thus by Part \ref{lem:atrans:preim} of Lemma \ref{lem:atrans}, we can apply $\eta_a$ to get $[\alpha(V'),\alpha(U') \cup \{k\}] = [V',U' \cup \{k\}]$. This has type $(v,u+1)$ as required.
\qed
\end{proof}
\subsection{Proof of Main Theorem}
Our main theorem, restated below, now follows easily:
\mainthm*

\begin{proof}
Since $L$ has syntactic complexity $n^n$, Lemma \ref{lem:reach} holds for minimal DFAs of atoms of $L$. It was shown in~\cite{BrTa12} that if these strong-connectedness and reachability results hold, the number of reachable intervals in the minimal DFA of an atom of $L$ is equal to the maximum possible quotient complexity of the atom. Hence these results suffice to establish that each atom has maximal quotient complexity. \qed
\end{proof}

\section{Conclusions}
Maximal quotient complexity of atoms defines a new complexity class of regular languages.
We have related this new measure to syntactic complexity and quotient complexity of reversal.  Such relations are important, since they often make it  possible to avoid proofs of complexity results implied by other known complexity results. We believe that this subject deserves further study.

\end{document}